\documentclass[11pt, dvipsnames]{article}

\pdfoutput=1
\usepackage{latex/macros}
\usepackage{latex/header}
\usepackage{latex/optics}
\usepackage{float}
\allowdisplaybreaks
\usepackage{combdiagrams}
\newcommand\delay{\mathsf{d}}
\newcommand\feedback{\mathsf{Fbk}}
\newcommand\Circ{\mathit{Circ}}
\newcommand\Learner{\mathit{Learner}}
\author{Mario Román}
\date{\today}
\title{Comb Diagrams for Discrete-Time Feedback}
\hypersetup{
 pdfauthor={Mario Román},
 pdftitle={Comb Diagrams for Discrete-Time Feedback},
 pdfkeywords={},
 pdfsubject={},
 pdfcreator={Emacs 27.0.50 (Org mode 9.1.14)}, 
 pdflang={English}}
\begin{document}

\maketitle

\begin{abstract}
The data for many useful bidirectional constructions in applied category theory (optics, learners, games, quantum combs) can be expressed in terms of diagrams containing \emph{holes} or \emph{incomplete parts}, sometimes known as \emph{comb diagrams}. We give a possible formalization of what these circuits with \emph{incomplete parts} represent in terms of symmetric monoidal categories, using the dinaturality equivalence relations arising from a coend.  Our main idea is to extend this formal description to allow for infinite circuits with holes indexed by the natural numbers. We show how infinite combs over an arbitrary symmetric monoidal category form again a symmetric monoidal category where notions of \emph{delay} and \emph{feedback} can be considered.  The constructions presented here are still preliminary work.
\end{abstract}

\section{Introduction: finite combs}
\label{sec:orge507bd7}
The name ``comb diagram'' comes from the \emph{quantum combs} present on the work of Chiribella, D'Ariano and Perinotti \cite{chiribella08}, where they are defined as ``circuit boards in which one can insert variable subcircuits''. For our purposes, we will assume that each circuit \emph{board} can have multiple \emph{holes} waiting for the insertion of potentially different subcircuits. Consider the following circuit \cite[Figure 1]{chiribella08} and its adaptation.
\[
\includegraphics[scale=0.14]{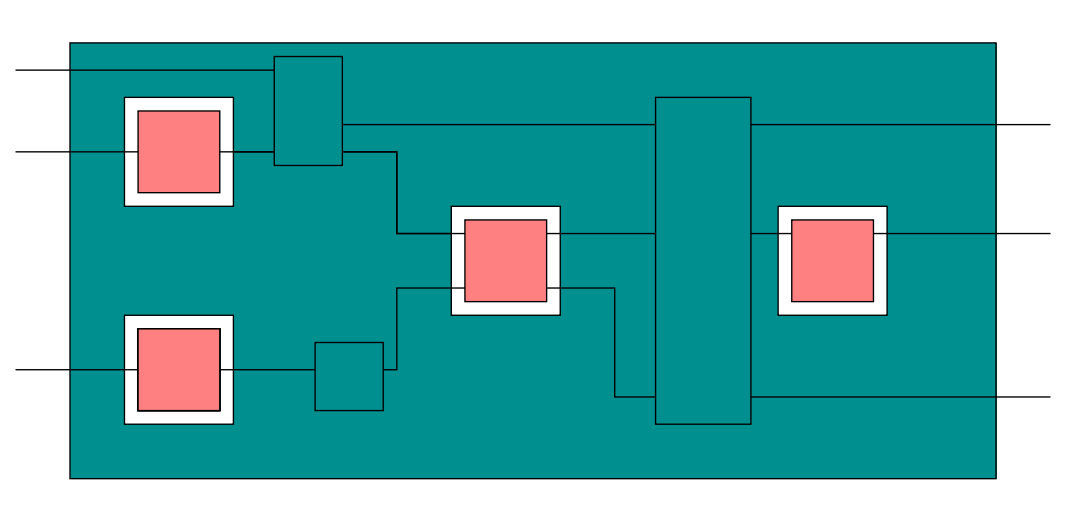}\qquad
\chiribellacircuit\]
Their original work does not mention category theory, but the category theorist will recognize this circuit as a diagram on a, possibly symmetric, monoidal category.  Rewording our goal, we want to consider \emph{holes} or \emph{incomplete parts} inside symmetric monoidal categories, while keeping a notion of equality between incomplete circuits that preserves the transformations of the graphical calculus of symmetric monoidal categories. A first simplification of the problem comes from the same article.
\begin{quote}
``\emph{It is clear that by reshuffling and stretching the internal wires any}
\emph{circuit board can be reshaped in the form of a ”comb”, with an ordered}
\emph{sequence of slots, each between two successive teeth, as in}
\emph{Fig. 3. The order of the slots is the causal order induced by the flow}
\emph{of quantum information in the circuit board.}'' -- \cite{chiribella08}
\end{quote}
Which can be translated for the category theorist as saying that we should use the symmetric structure to push the holes to the boundaries of the diagram where possible. The resulting diagram, adapted again for clarity, delineates a characteristic comb shape.
\[\chiribellasimplified{}\]
Our goal is to describe circuit boards like these. Assume we want to describe, for instance, the set of all the possible circuit boards with two holes.  That means we want to consider all possible triples of morphisms of the following shape.
\[\afourcomb{}\]
However, just defining combs to be a triple \((f,g,h)\) would miss the point, as it would not equate diagrams that are the same up to the usual transformations in a symmetric monoidal category. We need to keep track of the \emph{connecting wire} between morphisms; and we do this by explicitly quotienting out by an equivalence relation generated by all pairs of diagrams with the following shapes.
\begin{align*}
& \afourright \\
\sim & \\
& \afourcombleft
\end{align*}
The crucial step is to observe that these quotient relations can be rewritten in a compact way using the dinaturality conditions of a \emph{coend}, as defined for instance in \cite[\S IX.6]{maclane78}. The advantage of this description is that coends follow some practical rules of manipulation based on the Yoneda lemma; this is the \emph{coend calculus} described by \cite{loregian15}. We can regard the same triple, this time as an element of a set described as a coend.
\begin{align*}
[f,g,h] \in \int^{M_0,\dots,M_2} 
\C(X_0, M_0 \otimes Y_0) & \times 
\C(X_1 \otimes M_1, Y_1 \otimes M_2) \times \C(X_2 \otimes M_2, Y_2).
\end{align*}
Following this idea, we can propose a definition of comb in terms of coends.

\begin{definition}
\label{def:comb}
An n-\textbf{comb} between two families of objects \(X_0,\dots,X_n \in \C\) and \(Y_0,\dots,Y_n \in \C\)
is an element of the following set.
\[\Comb_n(X,Y) \coloneqq
\int^{M_0,\dots,M_{n-1}}
\prod_{i = 0}^{n}
\C(M_{i-1} \otimes X_i, M_i \otimes Y_i),
\mbox{ where } M_{-1},M_n \coloneqq I.\]
\end{definition}

\subsection{Contributions}
\label{sec:org9d0399e}

\begin{itemize}
\item A proposed definition of \(n\mbox{-comb}\) (Definition \ref{def:comb}) and \(\infty\mbox{-comb}\) (Definition \ref{def:omegacomb}) in terms of coends, as a way to model \emph{holes in monoidal categories}. The definition seems to correspond with the intuition in the literature and relate to other similar domain-specific definitions such as quantum combs (\S \ref{sec:causal}), lenses (\S \ref{sec:lenses}), and learners (\S \ref{sec:learners}). A discussion on the graphical representation of the quotienting of a coend in terms of monoidal categories, inspired by Riley's notation for optics \cite[\S 2]{riley18}.

\item A symmetric monoidal category of \emph{\(\infty\)-comb diagrams} (\S \ref{sec:categoryinftycombs}) over an arbitrary symmetric monoidal category, together with definitions for a \emph{delay} functor and a \emph{feedback} operator (\S \ref{sec:feedback}). We show that, in the cartesian case, this category particularizes into a Kleisli category for a suitable comonad (\S \ref{sec:cartesian}).
\end{itemize}

\section{Infinite combs}
\label{sec:orgc760b45}
\subsection{Motivation and examples}
\label{sec:org8ae4cef}
Morphisms in a monoidal category are usually interpreted as processes; and composition \((g \circ f)\) is interpreted as the process \(f\) happening strictly after the process \(g\).  Finite combs for a fixed \(n \in \bN\) are processes taking \(n\) inputs and producing \(n\) outputs, but they do so in a very specific order: they take the first input, \(X_0\), and produce the first output \(Y_0\), only after producing this first output they take the second input, \(X_1\) and produce the second output, \(Y_1\), only after producing this second output they take the third input, \(X_2\), and so on.  This distinguishes them from morphisms \(\otimes_{i=0}^n X_i \to \otimes_{i=0}^n Y_i\).

In practice, processes that alternatively take intputs and produce outptus do so in a indeterminate number of stages; they could even be processes that do not terminate.  From this perspective, finite combs feel too limited. The purpose of this section is to construct a category where morphisms are circuits that take inputs and produce outputs in a potentially infinite number of stages.   The category itself will turn out to be again a symmetric monoidal category that is suitable to define discrete dynamical systems. We anticipate it with two examples.

\begin{exampleth}[Fibonacci sequence]
\label{ex:fibonacci}
Consider the bialgebra of natural numbers \(\bN\) with copying \((\comultiplication) \colon \bN \to \bN \times \bN\), discarding \((\counit) \colon \bN \to 1\), addition \((\multiplication) \colon \bN \times \bN \to \bN\) and zero \((\unit)\colon 1 \to \bN\). This is the setting for \emph{Graphical linear algebra} \cite{bonchi:diagrammaticalgebra}.  The following diagram represents a morphism that computes the Fibonacci sequence.
\[\fibonacci
\]
There are many details to unpack here. This diagram is describing a state \(\Comb_{\infty}(1,\bN)\) in a category of \(\infty\)-combs. States \(\Comb_{\infty}(1,\bN)\) happen to correspond to infinite lists of elements of \(\bN\) (see \S \ref{sec:states}) and, in particular, this state corresponds to the Fibonacci sequence \([0,1,1,2,3,5,\dots]\).
The two markings for the initial values (\(0\) and \(1\)) are morphisms in the category of \(\infty\)-combs, and they are required to make compositions well-typed. The feedback operator is \emph{not} a trace, but we will describe it more carefully in \S \ref{sec:feedback}. Finally, the diagram can be \emph{unfolded} into the \(\infty\)-comb it represents.
\begin{align*}
& \unfoldfibonacci \\
= & \\
& \simplifyfibonacci \\
\end{align*}
\end{exampleth}

\begin{exampleth}[Probabilistic dynamical system]
\label{ex:dynamical}
The construction can be repeated in arbitrary symmetric monoidal categories that are not necessarily cartesian. Let \(D\) be the finite distribution monad. Consider a probabilistic discrete version of the Lotka-Volterra equations (also known as ``predator-prey'') with two initial populations of rabbits, \(r_0\), and foxes, \(f_0\).  These populations evolve in discrete time according to some probabilistic functions \(p \colon R \times F \to DR\) and \(q \colon R \times F \to DF\) that take both populations as inputs. The following is a valid diagram in the category of \(\infty\)-combs over the Kleisli category of the distribution monad, \(\operatorname{Kl}(D)\), describing how both populations interact over time.
\[\rabbits\]
The diagram is actually describing a state \(\Comb_{\infty}^{\operatorname{Kl}(D)}(1, R \times F)\), which happens to correspond to a coherent family of distributions, \(\varprojlim_n D(\prod^n_{i=0} R \times F)\). The unfolded comb is an \emph{incomplete circuit} in the Kleisli category of the distribution monad.
\[\unfoldrabbit\]
\end{exampleth}

\subsection{\(\infty\)-combs}
\label{sec:org1f9153a}
We will construct infinite combs as elements of the inverse limit of some incomplete n-combs. Let us first introduce some notation.  We use the \emph{ground} symbol (\(\ground\)) to denote the fact that we consider diagrams of that shape quotiented by the equivalence relation that disregards any morphism on that wire. For instance, the diagram
\[\begin{aligned}
\exampleground  &
\quad\mbox{denotes diagrams of the form} &
\examplenoground &
\quad\mbox{quotiented by} &
\examplewithbox \sim \examplenoground\ .
\end{aligned}\]
Secondly, we use holes (\(\hole{}\)) in our diagrams, as in the previous section. They are to be understood formally
as pairs of diagrams with some wires connected, quotiented by the equivalence relation that \emph{slides morphisms along wires}.  For instance, the diagram
\[\begin{aligned}
\examplecomb &
\quad\mbox{denotes a pair of morphisms} &&
\leftsquare \ \rightsquare \\
& \quad\mbox{quotiented by} &&
\examplewithbox \ \rightsquare \sim \leftsquare \ \examplewithboxright\ .
\end{aligned}\]
Formally, we shall make use of the dinaturality condition of a coend or a colimit. For instance,
\[\begin{aligned} 
\left( \combformally \right) &\coloneqq& [f] &\in \int^{M \in \C} \C(X , M \otimes Y); \\
\left( \holeformally \right) &\coloneqq& [f_0,f_1] &\in \int^{M \in \C} \C(X_0 , M \otimes Y_0) \otimes \C(M \otimes X_1, Y_1).
\end{aligned}\]
We will obtain our infinite comb diagrams as an inverse limit of finite diagrams. In
order to achieve this, we start by defining what a \emph{diagram until stage} \(n\) looks like for any
arbitrary \(n \in \bN\). We will call \(\Comb_n^{+}\) to the set of diagrams with this shape. They will be different
from the previously defined n-combs (Definition \ref{def:comb}) in that they will leave a wire \emph{open} to be connected to the
next stage. For instance, we
define
\begin{align*}    
\Comb_0^{ + } \coloneqq \left\{ \combone\right\},\ 
\Comb_1^{ + } \coloneqq \left\{ \combtwo\right\},\ 
\Comb_2^{ + } \coloneqq \left\{ \combthree\right\},  
\dots
\end{align*}
and we construct maps \(\Comb^{n+1} \to \Comb^{n}\) by projecting the \(n\) first components. Note
how the quotient conditions we defined previously are necessary to make this maps
well-defined. Formally, given two numerable families of objects \(X,Y \in [\bN,\C]\),
we are defining a set of \(n\mbox{-combs}^+\), with an open wire, for every \(n \in \bN\).
\[
\Comb_n^{+}(X,Y) \coloneqq
\int^{M_0,\dots,M_n}
\prod_{i = 0}^n
\C(M_{i-1} \otimes X_i, M_i \otimes Y_i),
\mbox{ where } M_{-1} \coloneqq I.
\]
Our definition of infinite comb, \(\Comb_{\infty}\), will be as the inverse limit of a chain
\[
\Comb_0^{ + } \gets \Comb_1^{ + } \gets \Comb_2^{ + } \gets \dots
\]
where the morphisms \(\Comb_{n+1}^{+} \to \Comb_n^{ + }\) are projections. Note the importance of
the quotienting to make these maps well-defined.

\begin{definition}
\label{def:omegacomb}
An \(\infty\)-comb between two families of objects \(X,Y \in [\bN, \C]\) is an element of the inverse limit
\[
\Comb_{\infty}(X,Y) \coloneqq
\varprojlim_n \int^{M_0,\dots,M_n}
\prod_{i = 0}^n
\C(M_{i-1} \otimes X_i, M_i \otimes Y_i),
\mbox{ where } M_{-1} \coloneqq I.\]
Combs in this sense are to be seen as sequences of morphisms \(f \coloneqq [f_0, f_1, \dots ]\) quotiented by an equivalence relation that equates
\begin{align*}
[ \dots, (m_{i-1} \otimes \id) \circ f_{i-1}, (m_i \otimes \id) \circ f_i, (m_{i+1} \otimes \id) \circ f_{i+1}, \dots] & \sim \\
[ \dots, f_{i-1} \circ (m_{i-2} \otimes \id), f_i \circ (m_{i-1} \otimes \id) , f_{i+1} \circ (m_i \otimes \id), \dots],
\end{align*}
for every family of \(m_1,m_2,\dots\) suitably typed in \(\C\). Let us introduce diagramatic notation for these morphisms.  A generic morphism \(f \in \Comb_{\infty}(X,Y)\) will be written from now on as the following diagram.
\[\genericmorphism\]
\end{definition}

\begin{remark}
\label{remark:semicartesian}
We could have also defined \(n\mbox{-combs}^+\) in terms of \(n\mbox{-combs}\) as
\[\Comb_n^{+}(\{X_0,\dots,X_n\},\left\{ Y_0,\dots,Y_n \right\}) \coloneqq \int^{M_n} \Comb_n(\{X_0,\dots,X_n\},\left\{ Y_0,\dots,Y_n \otimes M_n \right\}).
\]
In fact, when \(\C\) is a \emph{semicartesian category}, meaning that the monoidal unit is a
terminal object, both definitions coincide because of the Yoneda lemma.
\begin{align*}
& \int^{M_n} \Comb_n(\{X_0,\dots,X_n\},\left\{ Y_0,\dots,Y_n \otimes M_n \right\})\\
\cong & \quad \mbox{(The unit is terminal)} \\
& \int^{M_n} \Comb_n(\{X_0,\dots,X_n\},\left\{ Y_0,\dots,Y_n \otimes M_n \right\}) \times \C(M_n, I) \\
\cong & \quad \mbox{(Yoneda lemma)} \\
& \Comb_n(\{X_0,\dots,X_n\},\left\{ Y_0,\dots,Y_n \right\}).
\end{align*}
For semicartesian categories, the definition of \(\infty\)-combs can be rewritten as follows.
\[\Comb_{\infty}(X,Y) \coloneqq
\varprojlim_n \int^{M_0,\dots,M_{n-1}}
\prod_{i = 0}^{n-1}
\C(M_{i-1} \otimes X_i, M_i \otimes Y_i),
\mbox{ where } M_{-1},M_n \coloneqq I.\]
\end{remark}

\subsection{The symmetric monoidal category of \(\infty\)-combs}
\label{sec:org4bda100}
\label{sec:categoryinftycombs}
The category \([\bN,\C]\) is symmetric monoidal with the structure inherited from applying the monoidal product of \(\C\) pointwise. This, in turn, will induce a symmetric monoidal structure on \(\Comb_\infty\). Given two \(\infty\)-combs \(f \in \Comb_\infty(X, Y)\) and \(g \in \Comb_\infty(Y, Z)\), we can sequentially compose them into \((g \circ f) \in \Comb_\infty(X, Z)\), as in the following diagram.
\[\sequentialcomposition\]
Given two \(\infty\)-combs \(f \colon X \to Y\) and \(g \colon X' \to Y'\), we can compose them in parallel into a comb \((f \otimes g) \colon X \otimes X' \to Y \otimes Y'\), as in the following diagram.
\[\parallelcomposition{}\]
Moreover, we can lift a family of morphisms \(f_n \colon X_n \to Y_n\) to the following comb. This will later define an identity-on-objects functor \(i \colon [\bN,\C] \to \Comb_\infty\).
\[\inclusiondiagram\]

\begin{proposition}
The previous data determines a symmetric monoidal category \(\Comb_\infty\) with a strict monoidal identity-on-objects functor \([\bN,\C] \to \Comb_{\infty}\).
\end{proposition}
\begin{proof}
Let us start by showing that the sequential composition previously defined
is indeed associative.  In fact, the following two diagrams represent the 
same comb.
\begin{align*}
& \assocone \\
= & \\
& \assoctwo
\end{align*}
The identity \(\infty\)-comb can be lifted from the identity in \([\bN,\C]\), and it can be checked to be the unit of composition. The same can be done  with the unitors and associators, as the monoidal product coincides on objects; checking that they satisfy the required axioms is straigthforward in the graphical calculus.
\end{proof}

\subsection{Delay and feedback}
\label{sec:orgefae7a6}
\label{sec:feedback}
There exists a fully-faithful and strong monoidal functor \(\delay \colon [\bN,\C] \to [\bN,\C]\) that \emph{shifts by one} every sequence of objects, defined as \(\delay(X)_n \coloneqq X_{n+1}\) for \(X \in [\bN,\C]\).  The lifting of this functor to the category of combs is what we will call the \emph{delay functor} \(\delay \colon \Comb_{\infty} \to \Comb_{\infty}\). Given any \(f \in \Comb_\infty(X,Y)\), we can define \(\delay f \in \Comb_\infty(\delay X, \delay Y)\) as in the following diagram, making the first morphism be the empty diagram. This assignment can be shown to be functorial.
\[\delayedcomb\]

\begin{definition}
The \emph{feedback} operator \(\feedback^X \colon \Comb_\infty(\delay X \otimes A , X \otimes B) \to \Comb_\infty(A,B)\) is defined as sending the following generic \(\infty\)-comb \(f \in \Comb_\infty(\delay X \otimes A , X \otimes B)\),
\[\beforefeedback\]
to the following \(\infty\)-comb \(\feedback^X(f) \in \Comb_\infty(A,B)\).
\[\afterfeedback\]
This feedback operator enjoys a trace-like property, in the sense that for every \(f \in \Comb_\infty(\delay Y \otimes A , X \otimes B)\) and \(g \in \Comb_\infty(X,Y)\), it holds that \(\feedback^Y((g \otimes \id) \circ f) = \feedback^X(f \circ (\delay g \otimes \id))\). After
applying convenient swappings of the wires, proving this equation amounts to check 
that the morphisms representing \(g\) can be slided past the holes.
\[\tracelikediagram\]
However, the absence of the yanking equation and the requirement for the start of the trace to be on the image of the delay functor clash with the axioms of a trace. We employ a graphical calculus similar to the one for spherical traced categories \cite[\S 4.5.3]{selinger10} in the examples, with the important caveat that the type of the feedback operator (\(\tinyfeedback\)) does not coincide with that of the trace, and considering that it does not satisfy the same graphical equations of a trace. We use this calculus in Examples \ref{ex:fibonacci} and \ref{ex:dynamical}.
\[ \ffeedback\ \coloneqq\ \feedback(f){{{{}}}}\]
Apart from naturality, that can be stated graphically in the same way as for traces; the property we have shown before amounts to the following graphical equation.
\[ \tracelikeleft =  \tracelikeright \]
\end{definition}

\subsection{States}
\label{sec:org802f69c}
\label{sec:states}
Let us describe what states \(\Comb_{\infty}(I,Y)\) are in the category of \(\infty\)-combs. Note first that the monoidal unit on the category of combs is the monoidal unit on natural-number-indexed objects \(I \in [\bN,\C]\),  which is in turn the constant family of objects given by the unit.

\begin{proposition}
For \(\C\) an arbitrary symmetric monoidal category, \(\Comb_n(I,Y) \cong \C\left( I , \otimes_{i=0}^{n} Y_i  \right)\), and
\[ \Comb_n^{+}(I,Y) \cong \int^{M} \C\left( I , (\otimes_{i=0}^{n} Y_i) \otimes M  \right).\]
As a consequence, we have a description of states in \(\Comb_{\infty}\).
\[ \Comb_\infty(I,Y) \cong \varprojlim_n \int^M \C\left( I , (\otimes_{i=0}^{n} Y_i) \otimes M \right).\]
\end{proposition}
\begin{proof}
We shall apply induction again. In the case, \(n=0\) both sides of the isomorphism are equal. In the case \(n+1\) we can see that
\begin{align*}
& \Comb_{n+1}(I,\left\{ Y_0,\dots,Y_{n+1} \right\}) \\
\cong & \quad \mbox{(Definition)} \\
& \int^{M_{n}} \Comb_{n}(I,\left\{ Y_0,\dots,Y_{n} \otimes M_n \right\}) \times \C(M_n , Y_{n+1})\\
\cong & \quad \mbox{(Induction hypothesis)} \\
& \int^{M_n} \C\left( I , (\otimes_{i=0}^{n} Y_i) \otimes M_n  \right) \times \C(M_n , Y_{n+1}) \\
\cong & \quad \mbox{(Yoneda)} \\
& \C\left( I , \otimes_{i=0}^{n+1} Y_i  \right).
\end{align*}
Finally, note that \(\Comb_n^{+}(I,Y) \cong \int^{M_{n}}\Comb_n(I,\left\{ Y_0,\dots, Y_n \otimes M_n \right\})\).
\end{proof}

\begin{remark}
The description of \(\Comb_\infty(I,Y)\) can be made more concrete in the case where the category \(\C\) is semicartesian. In this case, n-combs coincide with n-combs\(^{\text{+}}\), see Remark \ref{remark:semicartesian}, which makes
\[ \Comb_\infty(I,Y) \cong \varprojlim_n \C\left( I , \otimes_{i=0}^{n} Y_i \right).\]
Finally, when the semicartesian category has the required limit, these states in the category of combs can be rewritten simply as states of type \(\varprojlim_n \otimes_{i=0}^{n} Y_i\).
\end{remark}

\section{Cartesian infinite combs}
\label{sec:org3d85858}
\label{sec:cartesian}
Cartesian \(\infty\)-combs are interesting because of their simplified structure, which helps intuition with their monoidal counterparts.  We will characterize \(\infty\)-combs in a cartesian category as Kleisli morphisms for a comonad. Let us first characterize finite cartesian combs.

\begin{lemma}
\label{lemma:cartesianncomb}
Let \(\C\) be a cartesian monoidal category. Let \(X, Y \in \C\).
\[
\Comb^{+}_n(X,Y) \cong \Comb_n(X,Y) \cong \prod_{i=0}^{n-1} \C(X_0 \times \dots \times X_i,Y_i).
\]
\end{lemma}
\begin{proof}
Recall the definition of \(\Comb_n^{+}\). The following isomorphism follows from continuity of the hom-functor.
\begin{align*}
\Comb_{n+1}^{+}(X,Y)
\cong
\int^{M_0 \in \C} \C(X_0,Y_0 \times M_0) \times 
\Comb_n^{+}(\left\{ M_0 \times X_1, X_2, \dots \right\}, \left\{ Y_1, Y_2,\dots \right\}).
\end{align*}
A similar isomorphism can be shown for \(\Comb_{n+1}\). The rest of the proof is a straightforward application of induction over the length of the comb and the Yoneda lemma. For the case \(n=0\), we have the following isomorphism because of the cartesian structure and the Yoneda lemma.
\begin{align*}
\int^{M_0} \C(X_0, M_0 \times Y_0) \cong
\int^{M_0} \C(X_0, M_0) \times \C(X_0, Y_0) \cong \C(X_0,Y_0).
\end{align*}
Finally, for the case \(n+1\), the induction hypothesis can be used in conjunction with the previous observation
to prove an isomorphism.
\begin{align*}
& \int^{M_0,\dots,M_{n}} \prod_{i=0}^{n} \C(M_{i-1} \times X_i, M_i \times Y_i) \\
\cong & \quad\mbox{(Previous observation)} \\
& \int^{M_0 \in \C} \C(X_0,Y_0 \times M_0) \times 
\Comb^{n}_{\C}(\left\{ M_0 \times X_1, X_2, \dots \right\}, \left\{ Y_1, Y_2,\dots \right\}) \\
\cong & \quad\mbox{(Induction hypothesis)} \\
& \int^{M_0 \in \C} \C(X_0,Y_0 \times M_0) \times 
\prod_{i=1}^{n} \C(M_0 \times X_1 \times \dots \times X_i,Y_i) \\
\cong & \quad\mbox{(Products split)} \\ 
&\int^{M_0 \in \C} \C(X_0,M_0) \times \C(X_0, Y_0) \times 
\prod_{i=1}^{n} \C(M_0 \times X_1 \times \dots \times X_i,Y_i) \\
\cong & \quad\mbox{(Yoneda lemma)} \\ 
& \C(X_0, Y_0) \times 
\prod_{i=1}^{n} \C(X_0 \times X_1 \times \dots \times X_i,Y_i). & \qedhere
\end{align*}
\end{proof}

\begin{proposition}
Let \(\C\) be a cartesian category. We can characterize \(\infty\)-combs in \(\C\) as
\[
\Comb_\infty(X,Y) \cong \prod_{n = 0}^{\infty} \C(X_0 \times \dots \times X_n , Y_n).
\]
\end{proposition}
\begin{proof}
After the application of Lemma \ref{lemma:cartesianncomb}, we only need to observe that the inverse limit of the following diagram is the desired product with the comb maps coinciding with the projections.
\[\begin{tikzcd}[column sep=small]
\C(X_0,Y_0) & 
{ \prod_{i=0}^1 \C(X_0\times \dots \times X_i,Y_i)} \lar &
{ \prod_{i=0}^2 \C(X_0\times \dots \times X_i,Y_i)} \lar &
\dots \lar 
\end{tikzcd}\qedhere\]
\end{proof}

\begin{remark}
After this characterization, it is straightforward to show that, for \(\C\) a cartesian category, \(\Comb_\infty\) is equivalent to the Kleisli category for a comonad \(\Theta \colon [\bN,\C] \to [\bN,\C]\) defined on objects as \(\Theta(X)_n \coloneqq X_0 \times \dots \times X_n\).
\end{remark}

\section{Related work}
\label{sec:org7c93fb8}
\subsection{Feedback, trace, and fixed-point semantics}
\label{sec:orgda257ea}
\label{sec:katis}
After writing this text, the author found a remarkable similarity between the ideas of \emph{delay} and \emph{feedback} explained here and the informative work of Katis, Sabadini and Walters on feedback and trace \cite{katis:feedback}.  Even the terminology coincides quite closely.  It seems plausible that we can link our construction to theirs, and that \(\infty\)-combs could be made a concrete example of a \emph{category with feedback} as defined there. However, a direct attempt will not work because of the type of our feedback operator, that requires a delay on the domain. Particularly relevant for us is also their \(\Circ(\C)\) construction, which is almost the data for a piece of an \(\infty\)-comb.

\begin{definition}
\cite[Definition 2.4]{katis:feedback} Let \(\C\) be a monoidal category. For any two objects \(X,Y \in \C\), we define
\[
\Circ_{\C}(X,Y) \coloneqq \int^{M \in \mathsf{Core}(\C)} \C(M \otimes X, M \otimes Y).
\]
\end{definition}

We can give \(\Circ_{\C}\) category structure. The requirement for the coend to be taken over \(\mathsf{Core}(\C)\), the maximal subgrupoid of \(\C\), instead of \(\C\), distinguishes this category from what we would have defined, by analogy, to be a piece of a comb.

\subsection{Quantum causal structures}
\label{sec:org34ac943}
\label{sec:causal}
Our original source of inspiration was the categorical treatment of combs of Kissinger and Uijlen \cite{uijlen17}, where they refer to \cite{chiribella08}.  However, our usage of combs, and our definition, seem slightly different.  In order to compare them, we can study two particular cases.

\begin{itemize}
\item \emph{1-Combs} in compact closed category \(\C\) are four-partite states, as in \cite[\S 2.1]{uijlen17}.  Fixing \(A,B,C,D \in \C\), we can compute
\begin{align*}
& \int^{M \in \C} \C(A, M \otimes B) \times \C(M \otimes C, D) \\
\cong &\quad\mbox{(Dual of $C$)} \\
& \int^{M \in \C} \C(A, B \otimes M) \times \C(M, C^{\ast} \otimes D) \\
\cong &\quad\mbox{(Yoneda lemma)} \\
& \C(A,  B \otimes C^{\ast} \otimes D) \\
\cong &\quad\mbox{(Dual of $A$)} \\
& \C(I, A^{\ast} \otimes B \otimes C^{\ast} \otimes D).
\end{align*}

\item The data for a comb (as in Definition \ref{def:comb}) in \((\C,\otimes,I,\multimap)\) a symmetric monoidal closed category does \emph{not} coincide with the definition of states typed by a comb in \cite[Definition 6.6]{uijlen17}.  This can be explained by the fact that combs as in \cite[Definition 6.6]{uijlen17} are just notation for morphisms in a \emph{precausal} category.

\begin{proposition}
Let \(\C\) be symmetric monoidal closed.
\[
   \Comb_n(X,Y) \cong \C(I,X_0 \multimap (X_1 \multimap \dots (X_{n-1} \multimap (X_n \multimap Y_n) \otimes Y_{n-1}) \otimes Y_{n-2} \dots ) \otimes Y_0).
   \]
\end{proposition}
\begin{proof}
We proceed by induction. For \(n=0\), we have \(\C(X_0,Y_0) \cong \C(I,X_0 \multimap Y_0)\). For the case \(n+1\),
we note that
\begin{align*}
& \Comb_{n+1}(X,Y) \\
\cong & \quad\mbox{(Continuity of the hom functor)} \\
& \int^{M_n} \Comb_{n}(\left\{ X_0,\dots,X_n \right\},\left\{ Y_0,\dots,Y_{n-1},M_{n} \otimes Y_n \right\}) \times \C(M_n \otimes X_{n+1}, Y_{n+1})\\
\cong & \quad\mbox{(Induction hypothesis)} \\
& \C(I,X_0 \multimap (X_1 \multimap \dots (X_n \multimap M_{n} \otimes Y_n) \otimes Y_{n-2} \dots ) \otimes Y_0) \times \C(M_n, X_{n+1}\multimap Y_{n+1})\\
\cong & \quad\mbox{(Yoneda lemma)} \\
& \C(I,X_0 \multimap (X_1 \multimap \dots (X_n \multimap (X_{n+1}\multimap Y_{n+1}) \otimes Y_n) \otimes Y_{n-2} \dots ) \otimes Y_0). \qquad \qedhere
\end{align*}
\end{proof}
\end{itemize}

This author, however, does not feel qualified to evaluate how the current construction relates or if it can be of any use to causal structures and prefers to refer the reader to the extensive work of Kissinger and Uijlen \cite{uijlen17} for a categorical treatment of these \emph{quantum combs}.

\subsection{Lenses and optics}
\label{sec:org5f20254}
\label{sec:lenses}
An inspiration for these diagrams and the treatment with coends is the lucid account of profunctor optics in functional programming by Riley \cite{riley18}.  The reader may notice that the data for the definition of an optic in a monoidal category coincides with that of a 1-comb; moreover, when discussing \emph{lawful optics} \cite[\S 3]{riley18}, Riley introduces notation that suggests the idea of 0-combs and 2-combs.

A popular example of optics are lenses, pairs of functions named \(\mathrm{view}\colon X_0 \to Y_0\) and \(\mathrm{update} \colon X_0 \times X_1 \to Y_1\). After the diagrams in \cite{riley18}, one can check that the data for a lens in a  cartesian category \(\C\) is exactly that of a 1-comb.
\begin{align*}
& \left\{ [f,g] \in \int^M\C(X_0, Y_0 \times M) \times \C(M \times X_1, Y_1)\ \middle|\ \lensstepone \right\}    \\
\cong & \quad\mbox{(There exists a unique split $f \coloneqq (f_1,f_2)$)} \\
& \left\{ (f_2, [f_1,g]) \in \C(X_0, Y_0) \times \int^{M} \C(X_0, M) \times \C(M \times X_1, Y_1) \ \middle|\ \lenssteptwo \right\} \\
\cong & \quad\mbox{(Yoneda reduction)} \\
& \left\{ (f_2,g \circ (f_1 \times \id)) \in \C(X_0, Y_0) \times \C(X_0 \times X_1, Y_1) \ \middle|\ \lensstepthree \right\}
\end{align*}
However, there is a crucial difference between optics (and in particular, lenses) and 1-combs. Optics come equipped with a given composition rule, namely, that of including one comb inside the other. In other words, the second part of a lens is contravariant. 1-Combs can be composed in at least two ways (see the following diagrams), and the composition of arbitrary finite combs admits a rich combinatorial structure of possible interleavings that we leave as further work.
\[
\combpositionone \quad\mbox{ vs. }\quad \combpositiontwo
\]
Related to this discussion, Spivak \cite[Example 2.5]{spivak19} observes that the data for a dynamical system coincides with that of lenses of a particular shape \((S,S) \to (A,B)\).  The combs we have described could maybe help to further justify this coincidence and the connection with wiring diagrams \cite{schultz16}.

\subsection{Learners}
\label{sec:orge5c3512}
\label{sec:learners}
The work of Fong and Johnson \cite{fong19} proposes a compositional approach to machine learning by exhibiting a monoidal category whose morphisms represent supervised learning algorithms.  A morphism in this category is given by a \emph{learner}, quotiented by a suitable equivalence relation.

\begin{definition}
\cite[Definition 4.1]{fong19} A \textbf{learner} taking inputs on a set \(A\) and producing outputs on a set \(B\) is given by

\begin{itemize}
\item a set of \emph{parameters} \(P\),
\item an \emph{implementation} function \(i \colon P \times A \to B\),
\item an \emph{update} function \(u \colon P \times A \times B \to P\), and
\item a \emph{request} function \(r \colon P \times A \times B \to A\).
\end{itemize}
\end{definition}

The advantage of this definition is that it is very close to our intuition of what a learner should be.  However, in the same article on optics as bidirectional data accessors \cite{riley18}, Riley notices a sharp alternative definition in terms of coends that can be generalized to arbitrary monoidal categories.

\begin{definition}
\cite{riley18} Let \((\C,\otimes,I)\) be a monoidal category. A \textbf{learner} taking inputs on \(A \in \C\) and producing outputs on \(B \in \C\) is an element of the following set represented as a coend.
\[
\Learner(A,B) \coloneqq \int^{P,Q \in \C}\C(P \otimes A, Q \otimes B) \times \C(Q \otimes B, P \otimes A).
\]
\end{definition}

This could be related to a piece of an infinite comb, but a naive embedding of learners into \(\infty\)-combs will fail to be functorial, again because of the contravariant nature of the second part of the learner.   In any case, it is interesting to note how the \(\Circ\) construction (in \S \ref{sec:katis} and \cite[Definition 2.4]{katis:feedback}) seems precisely to be a learner without the contravariant part.  In other words, the data for an element of \(\Circ(\C)\), before considering the necessary quotienting, is given by

\begin{itemize}
\item a set of \emph{parameters} \(P\), and
\item an \emph{implementation-update} function \(i \colon P \otimes A \to P \otimes B\).
\end{itemize}

This contrasts as a simpler description, but the absence of a contravariant part makes it conceptually different.

\section{Conclusions}
\label{sec:orgbe54193}

A category whose morphisms can be used to encode discrete dynamical systems can be constructed from the same ideas that give rise to quantum \emph{combs}, \emph{lenses} and \emph{learners}.  We have not yet related this idea to other notions of discrete dynamical system, nor to other notions of feedback, and thus this work is still at a very early stage.  However, the construction itself seems to be useful to describe examples in a wide range of categories; and it helps explain, in elementary terms, why lenses, learners, and discrete time dynamical systems should be related.  In the context of an increasing interest on optics, we may consider useful to take the time to describe this naive approach, if only to compare it with further developments.

\subsection{Further directions}
\label{sec:org86d9b0f}

\begin{itemize}
\item A crucial next step is to axiomatize the most important properties from this construction and study the universal property of this construction.  We probably would need to axiomatize the properties a fully-faithful strong monoidal pointed delay functor and a feedback operator.

\item We have defined infinite comb diagrams, but diagrams usually only open at the extremes. A naive notion of infinite diagram following the technique we have presented would degenerate into a discrete category due to the strong conditions on the quotient relation. Which other ways of defining infinite diagrams are avaliable? Related to this, the choice of \(\infty\) as a symbol is deliberately ambiguous; the naming scheme for these constructions should be decided after some generalization is proposed. Using the natural numbers as indexing set is purely motivated by our applications, but repeating the reasoning with different totally ordered sets, or even posets, seems promising.

\item A straightforward generalization restricts the category over which we take coends.  We do not need the \(M_1,M_2,\dots\) in the definition of comb to live on \(\C\), but on any category with a strong monoidal functor to \(\C\). Intuitively, this would limit the \emph{memory} or the communication of every process with its future self. Are there interesting applications that are modelled by this kind of limitation?

\item We hope that our diagrammatic description of the similarities and differences between lenses, combs, feedback and learners using coend calculus inspires and helps the intuition on their study. The trace-like feedback structure of the category of learners is mentioned by Fong, Spivak and Tuyéras \cite[\S 7.5]{tuyeras19}, together with the need of a construction that helps on the study of recurrent neural networks.  Can we apply infinite comb diagrams and their feedback operator to the study of recurrent neural networks?

\item Open games \cite{ghani18} make an extensive use of lenses to model the two-stage process of moving and receiving a utility.  How do open games compare to 2-combs? How to handle or rewrite the contravariant part of open games? Can we apply \(\infty\)-combs to the study of repeated games?

\item The structure of infinite combs naturally suggests the idea of dialogue.  In the field of Categorical Compositional Distributional models of meaning (DisCoCat), there is an ongoing proposal \cite{coecke:discocirc} of modelling sentence composition using wires of indefinite length representing how \emph{agents} expand across the \emph{dialogue}. Can we use \(\infty\)-combs to model dialoguing agents in DisCoCat?

\item Signal flow diagrams \cite{bonchi:signalflow}, as described by Bonchi, Sobociński and Zanasi, share properties with \(\infty\)-combs, and their right trace looks close to the feedback operator. In fact, Example \ref{ex:fibonacci} is a repetition of \cite[Example 7.3]{bonchi:signalflow}. What is the precise relation? Can we use \(\infty\)-combs to provide semantics of signal flow diagrams?
\end{itemize}

\section{Acknowledgements}
\label{sec:orgdc1def3}
The author first noticed a connection between combs and lenses thanks to an exposition of the work of Kissinger and Uijlen \cite{uijlen17} by Daphne Wang. The ideas that developed here took inspiration and benefited greatly from discussions with Jules Hedges and Edward Morehouse about the combinatorial structure of the composition of finite combs; and from discussions with Elena Di Lavore on repeated games.

Mario Román was supported by the European Union through the ESF funded Estonian IT Academy research measure (project 2014-2020.4.05.19-0001).

\bibliographystyle{alpha}
\bibliography{latex/bibliography}
\end{document}